\documentclass{conm-p-l}

\copyrightinfo{2010}{}

\usepackage{graphicx}
\usepackage{subfigure}

\usepackage{amssymb,amsmath,amsthm,amscd}

\newtheorem{theorem}{Theorem}[section]

\theoremstyle{definition}

\theoremstyle{remark}
\newtheorem{remark}[theorem]{Remark}

\numberwithin{equation}{section}

\newcommand{\e}{\epsilon}

\renewcommand{\k}{\kappa}
\newcommand{\ga}{\gamma}

\newcommand{\ra}{\rightarrow}
\newcommand{\al}{\alpha}

\newcommand{\pa}{\partial}

\newcommand{\la}{\lambda}

\newcommand{\om}{\omega}

\newcommand{\non}{\nonumber}

\newcommand{\hu}{\hat{u}}
\newcommand{\hv}{\hat{v}}

\begin{document}

\title[A Resolution of the Paradox of Enrichment]{A Resolution of the Paradox of Enrichment}

\author{Z. C. Feng}
\address{Department of Mechanical and Aerospace Engineering, 
University of Missouri, Columbia, MO 65211}
\email{fengf@missouri.edu}

\author{Y. Charles Li}
\address{Department of Mathematics, University of Missouri, 
Columbia, MO 65211, USA}
\email{liyan@missouri.edu}
\urladdr{http://www.math.missouri.edu/~cli}

\curraddr{}
\thanks{}

\subjclass{Primary 92; Secondary 35, 34, 37}
\date{}

\dedicatory{}


\begin{abstract}
The paradox of enrichment was observed by M. Rosenzweig \cite{Ros71} in
a class of predator-prey models. Two of the parameters in the models are crucial for the paradox. These two 
parameters are the prey's carrying capacity
and prey's half-saturation for predation. Intuitively, increasing the carrying capacity due to enrichment of 
the prey's environment should lead to a more stable predator-prey system. Analytically, it turns out that 
increasing the carrying capacity always leads to an unstable predator-prey system that is susceptible to 
extinction from environmental random perturbations. This is the so-called paradox of enrichment. Our resolution 
here rests upon a closer investigation on a dimensionless number $H$ formed from the carrying capacity and the 
prey's half-saturation. By recasting the models into dimensionless forms, the models are in fact governed by 
a few dimensionless numbers including $H$. The effects of the two parameters: carrying capacity and half-saturation
are incorporated into the number $H$. In fact, increasing the carrying capacity is equivalent (i.e. has the same 
effect on $H$) to decreasing the half-saturation which implies more aggressive predation. Since there is no 
paradox between more aggressive predation and instability of the predator-prey system, the paradox of enrichment 
is resolved. 

The so-called instability of the predator-prey system is characterized by the existence of a stable limit cycle in the 
phase plane, which gets closer and closer to the predator axis and prey axis. Due to random environmental 
perturbations, this can lead to extinction. We also further explore spatially dependent models for which the 
phase space is infinite dimensional. The spatially independent limit cycle which is generated by a Hopf bifurcation 
from an unstable steady state, is linearly stable in the infinite dimensional phase space. Numerical simulations 
indicate that the basin of attraction of the limit cycle is riddled. This shows that spatial perturbations can sometimes 
(neither always nor never) remove the paradox of enrichment near the limit cycle!
\end{abstract}

\maketitle

\section{Introduction}

The paradox of enrichment was first observed by M. Rosenzweig \cite{Ros71} in a class of mathematical predator-prey models.
Since then, there have been a lot of studies on the subject \cite{GR72}  \cite{AW96}  \cite{BSB07} \cite{CP01} \cite{Die07} 
\cite{FW86} \cite{GK99} \cite{JG05} \cite{Kir98} \cite{MN08} \cite{RGB08} \cite{Rie74} \cite{RC07} \cite{SB95} 
\cite{SP11} \cite{VKDM04}. These studies cover a wide spectrum of topics including invulnerable prey, unpalatable prey,
prey toxicity, induced defense, spatial inhomogeneity etc..
The paradox roughly says that in a predator-prey system, increasing the nutrition to the prey may 
lead to an extinction of both the prey and the predator. It is possible that the paradox is purely an artifact of the mathematical 
models, while in reality increasing the nutrition never leads to extinction. Our studies here totally focus upon the mathematical 
models themselves. We are not exploring the experimental aspect of the subject. As far as the original mathematical models 
\cite{Ros71} are concerned, we notice that the paradox can be resolved once the models are put into dimensionless forms. In 
dimensionless forms, the essential functions of control parameters can be revealed.

\section{The Predator-Prey Model}

The predator-prey model is as follows (a spatio-temporal extension of one of those in \cite{Ros71}),
\begin{eqnarray}
& & \frac{\pa U}{\pa T} - C_1 \frac{\pa V}{\pa X} \frac{\pa U}{\pa X} = D \frac{\pa^2 U}{\pa X^2} +\al U 
\left (1- \frac{U}{b} \right ) - \ga \frac{U}{U+h} V , \label{Mod1} \\
& & \frac{\pa V}{\pa T} + C_2 \frac{\pa U}{\pa X} \frac{\pa V}{\pa X} = D \frac{\pa^2 V}{\pa X^2} +
\left (\k \ga \frac{U}{U+h} - \mu \right ) V , \label{Mod2} 
\end{eqnarray}
where $U$ is the prey density, $V$ is the predator density, $T$ is the time coordinate, $X$ is the one-dimensional 
space coordinate, $C_1$ and $C_2$ are the coefficients of migration due to predation, $D$ is the spreading (diffusion) 
coefficient of the species (chosen to be the same for both predator and prey), $\al$ is the maximal per capita birth rate of 
the prey, $b$ is the carrying capacity of the prey from the nutrients, $h$ is the half-saturation prey density for 
predation, $\ga$ is the coefficient of the intensity of predation, $\k$ is the coefficient of food utilization of the 
predator, and $\mu$ is the mortality rate of the predator. The last two terms in equation (\ref{Mod1}), i.e. the prey 
birth term and the predation term, can take many different specific forms, but have the same characteristics as 
(\ref{Mod1}), which leads to the paradox of enrichment, see \cite{Ros71}. 

Equations (\ref{Mod1})-(\ref{Mod2}) can be rewritten in the following dimensionless form:
\begin{eqnarray}
& & \frac{\pa u}{\pa t} - c_1 \frac{\pa v}{\pa x} \frac{\pa u}{\pa x} =  \frac{\pa^2 u}{\pa x^2} + u
(1-u ) - \frac{u}{u+H} v , \label{Md1} \\
& & \frac{\pa v}{\pa t} + c_2 \frac{\pa u}{\pa x} \frac{\pa v}{\pa x} =  \frac{\pa^2 v}{\pa x^2} + k
\left (\frac{u}{u+H} - r \right ) v , \label{Md2} 
\end{eqnarray}
where $u = U/b$, $v= V \ga /(\al b)$, $t=\al T$, $x=X\sqrt{\al /D}$, and the dimensionless numbers are given by
\begin{equation}
H = \frac{h}{b}, \ r = \frac{\mu}{\k \ga}, \ k = \frac{\k \ga}{\al}, \ c_1 = C_1 \frac{\al b}{\ga D}, 
\ c_2 = C_2 \frac{b}{D}.
\label{dln}
\end{equation}
We name $H$: the capacity-predation number, and $r$: the mortality-food number. These two dimensionless numbers are 
crucial in our resolution of the paradox of enrichment. The spatial domain is chosen to be finite $x \in [0, L]$.
Three types of boundary conditions can be posed,
\begin{enumerate}
\item Neumann boundary condition, 
\[
\frac{\pa u}{\pa x}|_{x=0,L} = \frac{\pa v}{\pa x}|_{x=0,L} = 0;
\]
\item periodic boundary condition, $u$ and $v$ are periodic in $x$ with period $L$;
\item Dirichlet boundary condition,
\[
u|_{x=0,L} = v|_{x=0,L} = 0.
\]
\end{enumerate}
In the Dirichlet boundary condition case, the spatially uniform dynamics [$\pa_x =0$ in (\ref{Md1})-(\ref{Md2})]
is excluded. Thus the orginal paradox of enrichment for the uniform dynamics posed by M. Rosenzweig \cite{Ros71} 
is also excluded. 

\section{Formulation of the Paradox of Enrichment}

The paradox of enrichment was originally formulated by M. Rosenzweig \cite{Ros71} for the spatially uniform 
dynamics [$\pa_x =0$ in (\ref{Mod1})-(\ref{Mod2})]:
\begin{eqnarray}
& & \frac{d U}{d T}  = \al U \left (1- \frac{U}{b} \right ) - \ga \frac{U}{U+h} V , \label{ud1} \\
& & \frac{d V}{d T}  = \left (\k \ga \frac{U}{U+h} - \mu \right ) V . \label{ud2} 
\end{eqnarray}
The paradox focuses upon the linear stability of the steady state given by
\[
\k \ga \frac{U}{U+h} - \mu = 0, \ \al \left (1- \frac{U}{b} \right ) - \ga \frac{1}{U+h} V = 0 .
\]
It turns out that when other parameters are fixed, increasing $b$ leads to the loss of stability of this steady state,
in which case, a limit cycle attractor around the steady state is generated. As $b$ increases, the limit cycle gets 
closer and closer to the $V$-axis. That is, along the limit cycle attractor, the prey population $U$ decreases to
a very small value. Under the ecological random perturbations, $U$ can reach $0$, i.e. extinction of the prey. With 
the extinction of the prey, the predator will become extinct soon. On the other hand, increasing $b$ means increasing 
the carrying capacity of the prey, which can be implemented by increasing the prey's nutrients, i.e. enrichment of 
the prey's environment. Intuitively, increasing $b$ should enlarge the prey population and make it more robust from 
extinction. This is the paradox of enrichment. 

\section{The Resolution of the Paradox of Enrichment}

In order to resolve the paradox of enrichment, it is fundamental to rewrite the system (\ref{ud1})-(\ref{ud2}) in 
the dimensionless form [$\pa_x =0$ in (\ref{Md1})-(\ref{Md2})]:
\begin{eqnarray}
& & \frac{d u}{d t} = u(1-u ) - \frac{u}{u+H} v , \label{UMd1} \\
& & \frac{d v}{d t} = k\left (\frac{u}{u+H} - r \right ) v , \label{UMd2} 
\end{eqnarray}
and the key to the resolution is a complete understanding of the dimensionless capacity-predation number $H$.

First we need to understand the dynamics of (\ref{UMd1})-(\ref{UMd2}) in details. For all parameter values, there 
are two trivial steady states:
\[
(I). \ u_*=0, \ v_* = 0; \quad (II). \ u_*=1, \ v_* = 0.
\]
The steady state (I) is a saddle for all parameter values. The steady state (II) is a stable node when $H> \frac{1}{r} - 1$,
a saddle when $0<H< \frac{1}{r} - 1$ in which case, a nontrivial steady state is born. That is, when $0<H< \frac{1}{r} - 1$,
there is a nontrivial steady state which is our main focus:
\begin{equation}
u_*=\frac{r}{1-r}H, \quad  v_*= (1-u_*)(u_*+H); \label{fp}
\end{equation}
which is the intersection point of the parabola and the vertical line (Figure \ref{PV}):
\[
(P). \ v= -\left [ u-\frac{1}{2}(1-H) \right ]^2 + \left [\frac{1}{2}(1+H) \right ]^2, \quad (V). \ u = \frac{r}{1-r}H.
\]
\begin{figure}[ht] 
\centering
\includegraphics[width=4.5in,height=4.5in]{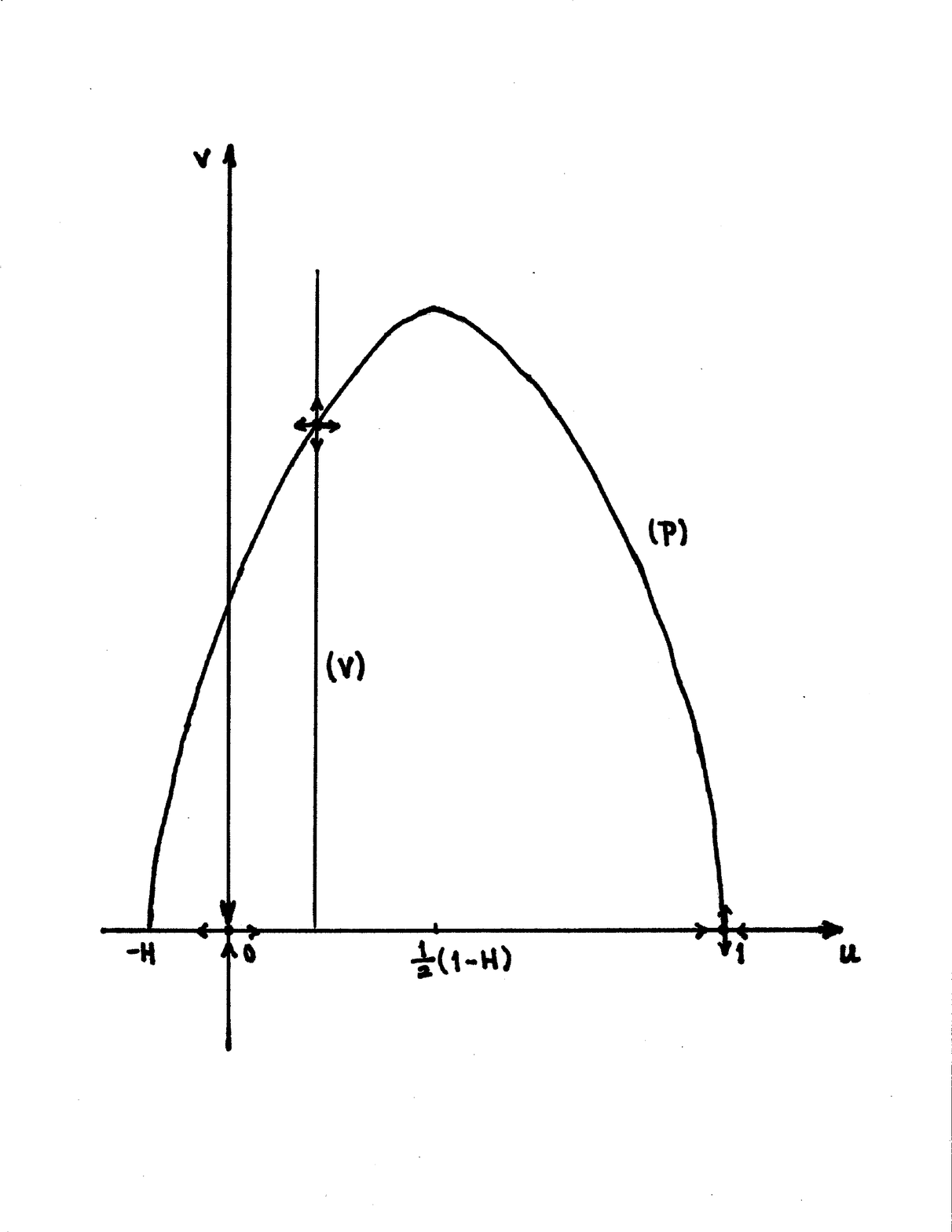}
\caption{The phase plane of the spatially uniform system (\ref{UMd1})-(\ref{UMd2}).}
\label{PV}
\end{figure}
\begin{figure}[ht] 
\centering
\includegraphics[width=4.5in,height=4.5in]{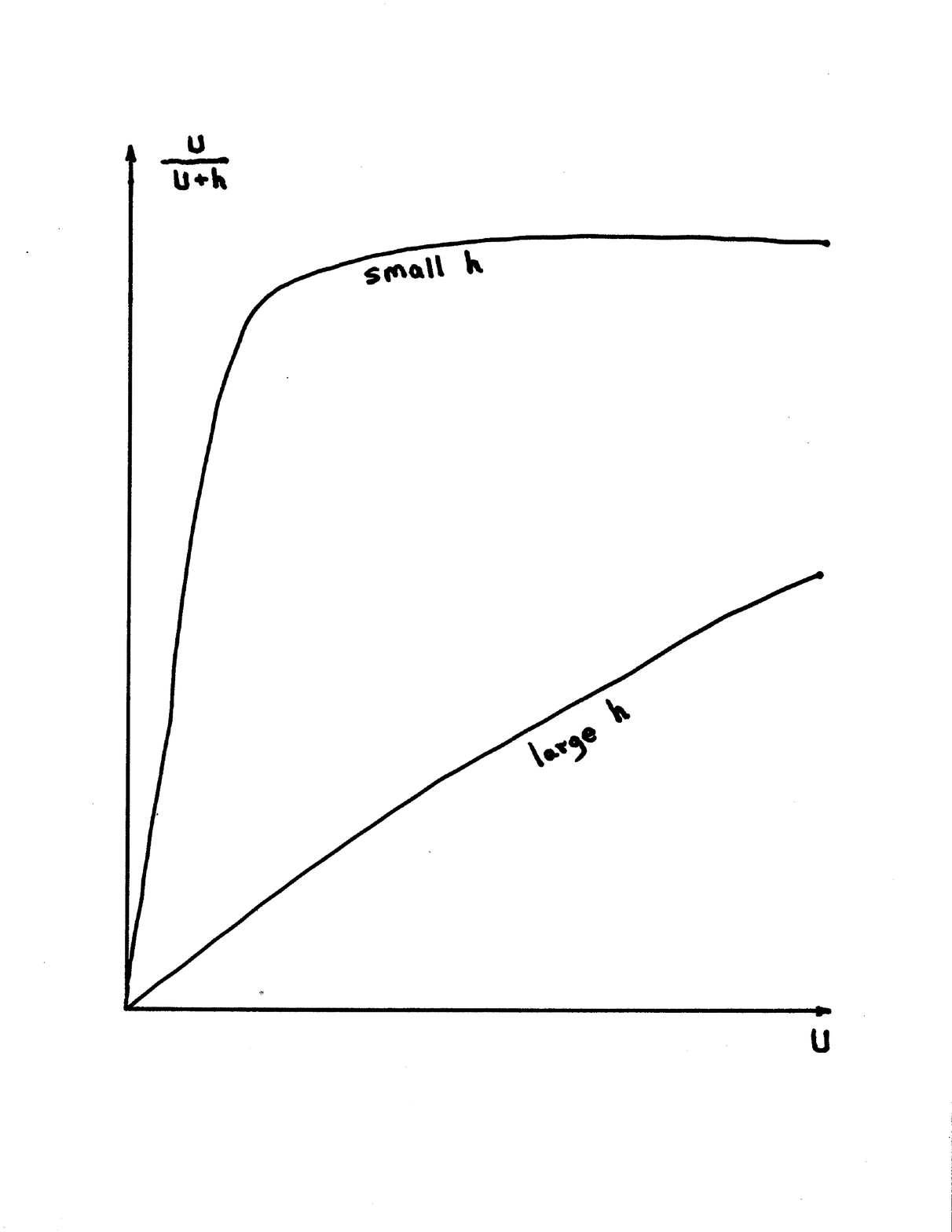}
\caption{The predation graphs.}
\label{slh}
\end{figure}
On the parabola (P), $\frac{d u}{d t} = 0$, and on the vertical line (V), $\frac{d v}{d t} = 0$. Linearizing the  
system (\ref{UMd1})-(\ref{UMd2}) at the steady state (\ref{fp}), we get 
\begin{eqnarray*}
& & \frac{d u}{d t} = \frac{u_*}{u_*+H} \left [ (1-2u_* -H)u - v \right ],  \\
& & \frac{d v}{d t} = \frac{kH(1-u_*)}{u_*+H} u .
\end{eqnarray*}
The eigenvalues of this linear system is given by
\begin{equation}
\la = \frac{1}{2} (1-2u_* -H) \pm \sqrt{ \left [ \frac{1}{2} (1-2u_* -H)\right ]^2 - kH \frac{1-u_*}{u_*}} .
\label{uev}
\end{equation}
The sign of the real part of $\la$ is decided by the sign of the quantity $1-2u_* -H$. Setting $1-2u_* -H =0$,
one gets 
\[
u_* = \frac{1}{2}(1-H),
\]
which is on the symmetry axis of the parabola (P). This leads to the following fact first observed by M. Rosenzweig 
\cite{Ros71}:
\begin{itemize}
\item If $u_*$ in (\ref{fp}) is to the left of the symmetry axis of the parabola (P), then the steady state (\ref{fp})
is linearly unstable. If $u_*$ in (\ref{fp}) is to the right of the symmetry axis of the parabola (P), then the steady 
state (\ref{fp})is linearly stable. If $u_*$ in (\ref{fp}) is on the symmetry axis of the parabola (P), then the eigenvalues
$\la$ in (\ref{uev}) are purely imaginary. 
\end{itemize}
Using this fact, we derive the following linear instability criterion of the steady state (\ref{fp}):
\begin{equation}
0<H<\frac{2}{1+r} - 1.
\label{H}
\end{equation}
(An interesting note is that the unstable zone (\ref{H}) is symmetric with respect to $H=r$, i.e. it is the same with 
$0<r<\frac{2}{1+H} - 1$.)
Based upon this instability criterion, we offer the following resolution of the paradox of enrichment.
\begin{itemize}
\item {\bf The Resolution}: Unlike the original form of the model (\ref{ud1})-(\ref{ud2}), the dimensionless form 
of the model (\ref{UMd1})-(\ref{UMd2}) is governed by $3$ dimensionless numbers $H$, $r$ and $k$ (\ref{dln}); while 
the instability of the steady state (\ref{fp}) is governed by $2$ of them, $H$ and $r$. $H$ is a ratio of the half-saturation 
$h$ and carrying capacity $b$, while $r$ is independent of $h$ and $b$, and is a ratio of predator mortality rate $\mu$ and 
coefficient of growth from food $\k \ga$. The instability criterion (\ref{H}) says that for a fixed $r$, $r \in (0,1)$;
when $H$ is smaller than $\frac{2}{1+r} - 1$, the steady state (\ref{fp}) is linearly unstable (leading to possible extinction).
The model displays a very special feature: Increasing the carrying capacity $b$ (for fixed half-saturation $h$) and 
decreasing the half-saturation $h$ (for fixed carrying capacity $b$) have the same effect on the capacity-predation number 
$H$, that is, $H$ decreases. Decreasing the half-saturation $h$ implies more aggressive predation (especially when the 
prey population $U$ is small), see Figure \ref{slh}. Notice that 
\[
\text{As } U \ra 0^+, \ \frac{U}{U+h} \ra 1;
\]
and 
\[
\frac{d}{dU} \frac{U}{U+h} \bigg |_{U=0} = 1/h .
\]
Since there is no paradox between more aggressive predation (especially when the prey population $U$ is small) and 
extinction of prey led by the instability of the steady state (\ref{fp}), the paradox of enrichment now reduces to a 
paradox between more aggressive predation (decreasing the half-saturation $h$) and enrichment (increasing the carrying 
capacity $b$). As mentioned above, the special feature of the model (\ref{UMd1})-(\ref{UMd2}) is that more aggressive 
predation (decreasing $h$) and enrichment (increasing $b$) is not a paradox, and results in the same effect on the governing 
dimensionless number $H$. This offers a resolution to the so-called paradox of enrichment. 
\end{itemize}

\section{More General Model}

From last section, we see that the predation term in (\ref{UMd1})-(\ref{UMd2}) is mostly responsible for generating 
the so-called paradox of enrichment. In this section, we will explore more general form of the predation term. Thus 
we will study the following more general model,
\begin{eqnarray}
& & \frac{d u}{d t} = u(1-u ) - f(bu) v , \label{gm1} \\
& & \frac{d v}{d t} = k\left (f(bu) - r \right ) v , \label{gm2} 
\end{eqnarray}
where $f$ is a monotonically increasing non-negative function. For the model (\ref{UMd1})-(\ref{UMd2}), $f(U)=\frac{U}{U+h}$. 
The nontrivial steady state for (\ref{gm1})-(\ref{gm2}) is given by 
\[
f(bu_*) = r, \ v_*= \frac{u_*(1-u_*)}{f(bu_*)} = \frac{1}{r}u_*(1-u_*),
\]
where $0<u_*<1$. Linearizing (\ref{gm1})-(\ref{gm2}) at this steady state, we get
\begin{eqnarray*}
& & \frac{d u}{d t} = \left [ 1-2u_* -\frac{b}{r} f'(bu_*)u_*(1-u_*)\right ] u - f(bu_*) v ,  \\
& & \frac{d v}{d t} = k\frac{b}{r} f'(bu_*)u_*(1-u_*)u .
\end{eqnarray*}
The eigenvalues of this linear system is given by
\begin{eqnarray}
\la &=& \frac{1}{2} \left [ 1-2u_* -\frac{b}{r} f'(bu_*)u_*(1-u_*)\right ] \non \\
& & + \sqrt{\left [\frac{1}{2} \left [ 1-2u_* -\frac{b}{r} f'(bu_*)u_*(1-u_*)\right ]\right ]^2 -
kbf'(bu_*)u_*(1-u_*)}. \label{geva}
\end{eqnarray}
The sign of the real part of $\la$ is decided by the sign of the quantity
\[
\frac{1}{2} \left [ 1-2u_* -\frac{b}{r} f'(bu_*)u_*(1-u_*)\right ].
\]
Besides $f(U)=\frac{U}{U+h}$, another natural model is $f(U)=U$. Notice that
\[
\frac{U}{U+h} \bigg |_{U=h} = 1/2, \quad U|_{U=h} = h;
\]
thus, for small $h$, the model $f(U)=\frac{U}{U+h}$ represents a much more aggressive predation than $f(U)=U$ when the 
prey population $U$ is small. As $U \ra +\infty$, 
\[
\frac{U}{U+h} \ra 1, \quad U \ra +\infty.
\]
That is, the model $f(U)=U$ represents unlimited predation ability for large prey population $U$. In this sense, 
$f(U)=\frac{U}{U+h}$ serves as a better model. On the other hand, the prey population is finite with capacity $b$. With 
a proper choice of the coefficient of predation $\ga$ (\ref{ud1}), $f(U)=U$ still represents a limited predation 
ability for prey population $U$ near its capacity. When $f(U)=U$, the eigenvalue (\ref{geva}) becomes
\[
\la = -\frac{r}{2b} \pm \sqrt{\left (\frac{r}{2b}\right )^2 - kr \left (1-\frac{r}{b}\right)},
\]
where $r<b$ is required. Thus for the model $f(U)=U$, the nontrivial steady state is always stable. 

\section{The Limit Cycle in the Phase Plane}

Returning to the spatially uniform system (\ref{UMd1})-(\ref{UMd2}), we can prove the following $\om$-limit set theorem.
\begin{theorem}
Under the dynamics of (\ref{UMd1})-(\ref{UMd2}), when $0<H<\frac{2}{1+r}-1$ ($0<r<1$), the $\om$-limit set of every point 
in the first open quadrant of the phase plane except the unstable steady state (\ref{fp}), is a periodic orbit 
(not necessarily the same periodic orbit). 
\end{theorem}
\begin{figure}[ht] 
\centering
\includegraphics[width=4.5in,height=4.5in]{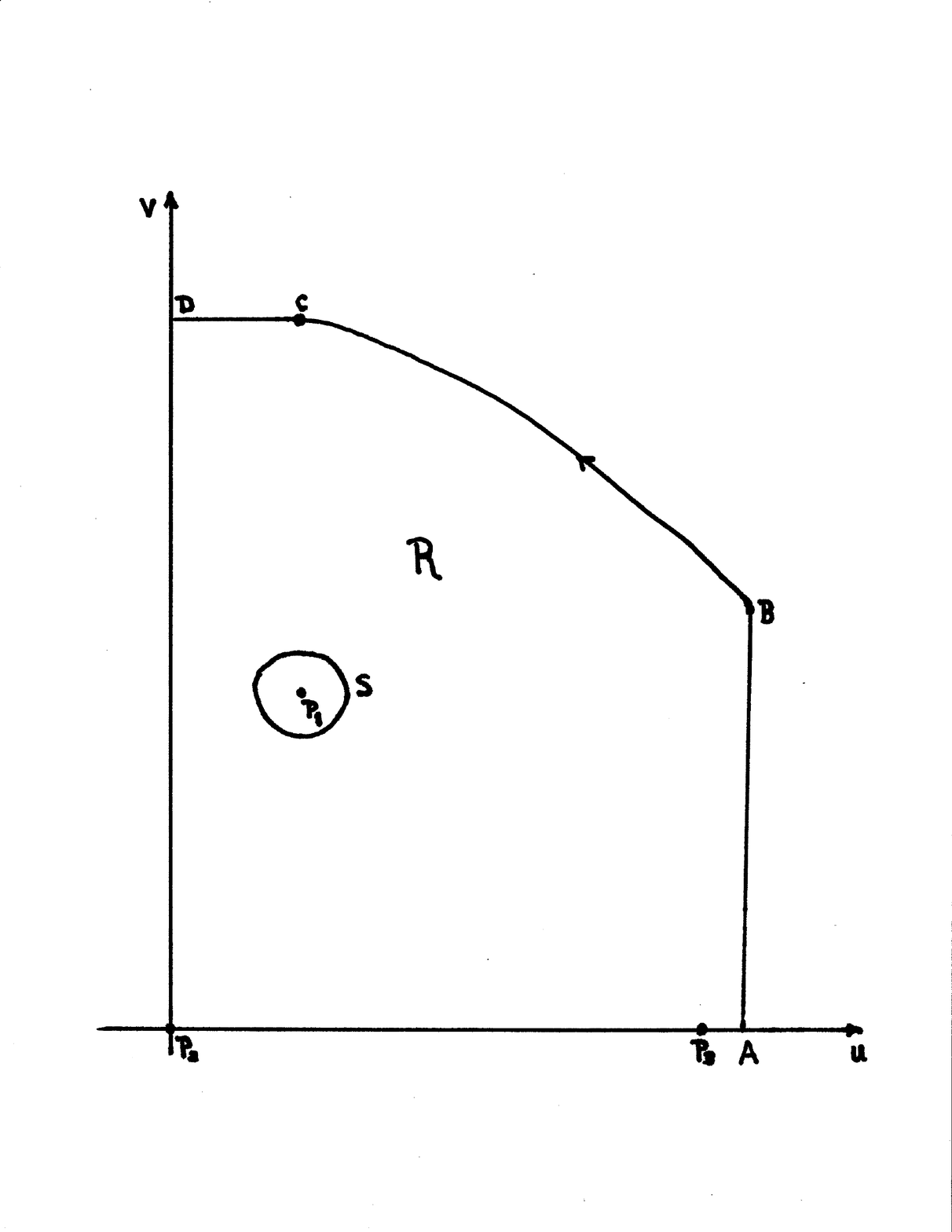}
\caption{The region setup for the proof of the existence of a limit cycle attractor.}
\label{lcf}
\end{figure}
\begin{figure}[ht] 
\centering
\includegraphics[width=4.5in,height=4.5in]{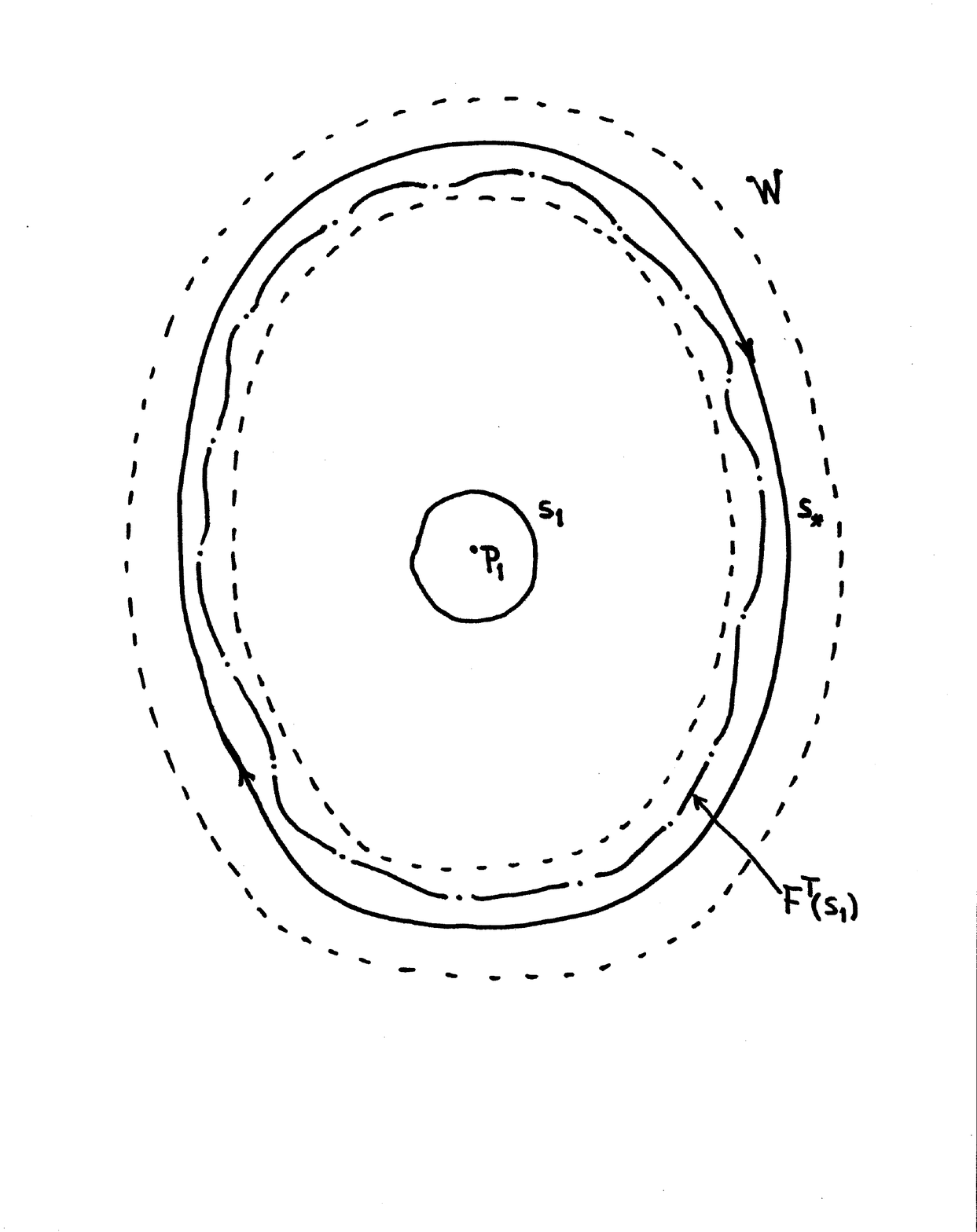}
\caption{The setup for the proof that the limit cycle loops around the steady state.}
\label{lca}
\end{figure}

\begin{figure}[ht] 
\centering
\includegraphics[width=4.5in,height=4.5in]{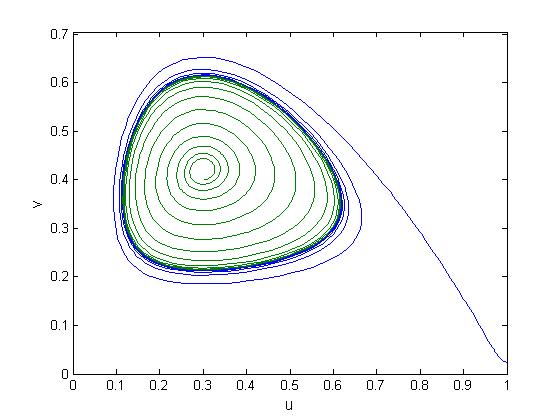}
\caption{The limit cycle attractor in the first open quadrant of the phase plane,
where $r=1/2$, $H=0.3$, $k=1$, and the two initial conditions are ($u=1.01, v=0.02$) and ($u=0.30, v=0.42$).}
\label{flc3}
\end{figure}

\begin{figure}[ht] 
\centering
\subfigure[$H=0.2$]{\includegraphics[width=2.3in,height=2.3in]{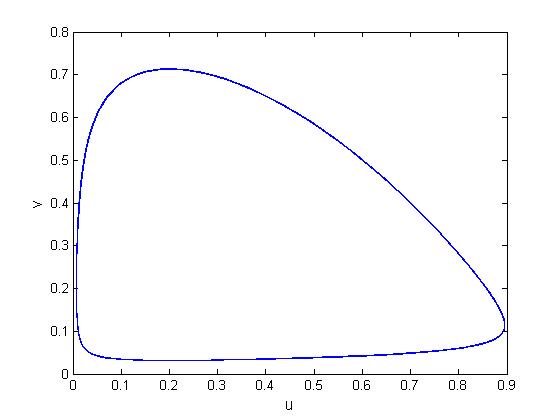}}
\subfigure[$H=0.1$]{\includegraphics[width=2.3in,height=2.3in]{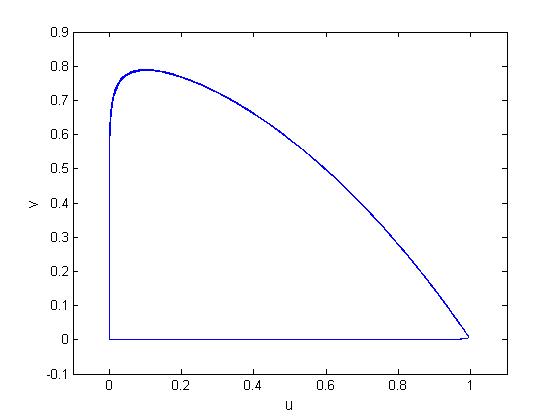}}
\caption{The deformation of the limit cycle as $H$ is decreased, where $r=1/2$ and $k=1$.}
\label{flc41}
\end{figure}

\begin{figure}[ht] 
\centering
\includegraphics[width=4.5in,height=4.5in]{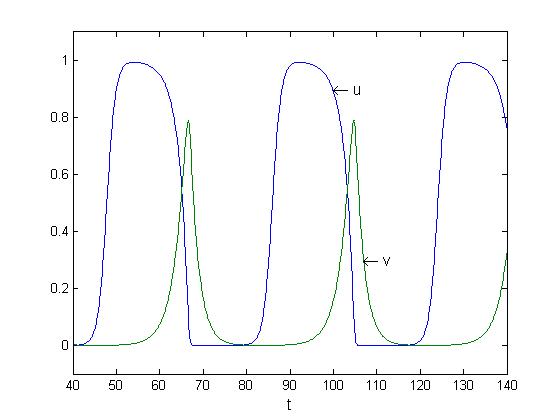}
\caption{The time series graph of the limit cycle in Figure \ref{flc41}(b), where $H=0.1$, $r=1/2$, and $k=1$.}
\label{flc2}
\end{figure}

\begin{figure}[ht] 
\centering
\includegraphics[width=4.5in,height=3.0in]{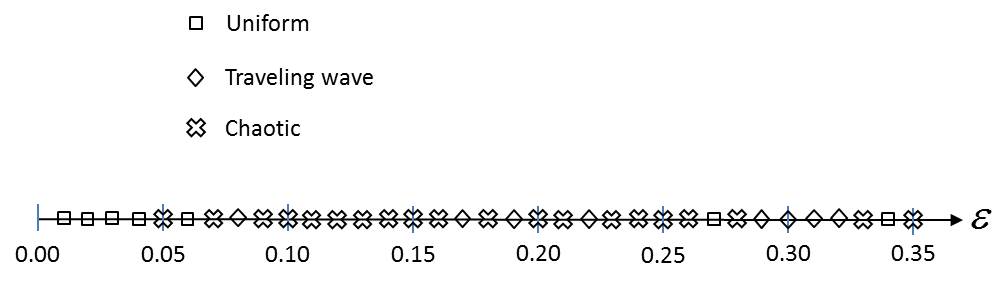}
\caption{An illustration of the riddled basin of attraction of the limit cycle on the plane, where the uniform, traveling wave,
and chaos asymptotic states are shown in Figures \ref{uni}, \ref{tw}, \ref{chaos}.}
\label{RB}
\end{figure}

\begin{figure}[ht] 
\centering
\includegraphics[width=4.5in,height=4.5in]{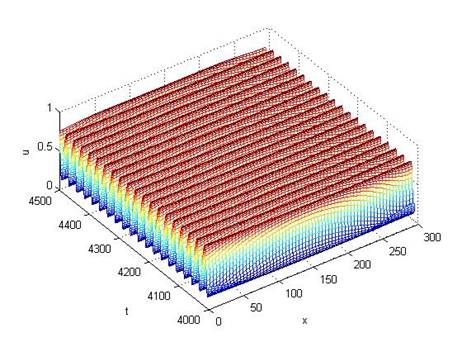}
\caption{The uniform asymptotic state referred to in Figure \ref{RB}.}
\label{uni}
\end{figure}

\begin{figure}[ht] 
\centering
\includegraphics[width=4.5in,height=4.5in]{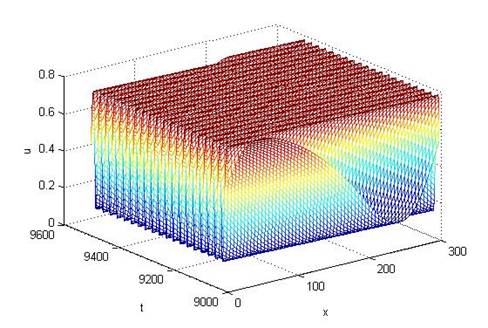}
\caption{The traveling wave asymptotic state referred to in Figure \ref{RB}.}
\label{tw}
\end{figure}

\begin{figure}[ht] 
\centering
\includegraphics[width=4.5in,height=4.5in]{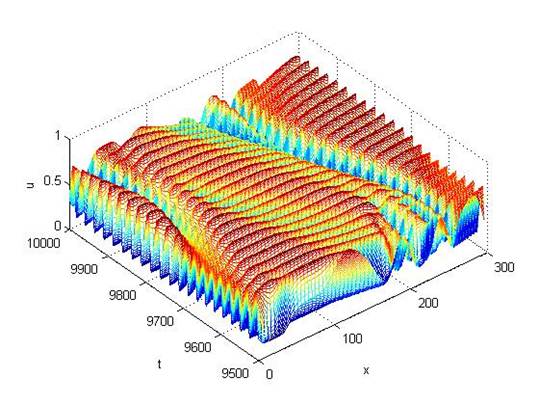}
\caption{The chaos asymptotic state referred to in Figure \ref{RB}.}
\label{chaos}
\end{figure}

\begin{proof}

First we set up the compact region $R$ as in Figure \ref{lcf}. $P_1$ is the steady state (\ref{fp}) which is unstable since 
$0<H<\frac{2}{1+r}-1$. $S$ is a small closed curve around $P_1$, on which the vector field given by the right hand side 
of (\ref{UMd1})-(\ref{UMd2}) is transversal to $S$ and points outside of $S$ since both eigenvalues (\ref{uev}) of the 
steady state (\ref{fp}) have positive real parts. $P_2$ is the unstable steady state ($0,0$). $P_3$ is the unstable steady 
state ($1,0$). $AB$ is a vertical line for which $u >1$, such that the vector field is transversal to $AB$ and points 
leftward by equation (\ref{UMd1}). The vertical coordinate of $B$ is chosen to be greater than $\frac{1}{4r}$, so that on 
the orbit $BC$, the right hand side of (\ref{UMd1}) is negative, where the horizontal coordinate of $C$ is the $u_*$ given 
in (\ref{fp}), i.e. $C$ and $P_1$ have the same horizontal coordinate. Notice that $\frac{u}{u+H}$ is strictly monotonically 
increasing in $u$. The right hand side of (\ref{UMd2}) at $C$ is zero. The right hand side of (\ref{UMd2}) on the orbit $BC$ 
(except at the point $C$) is positive. $CD$ is a horizontal segment on which the $v$ coordinate is a constant, and the right 
hand side of (\ref{UMd2}) is negative (except at the point $C$). The region $R$ is defined as the region outside $S$ and 
inside the loop $ABCDP_2P_3A$. The region $R$  is a positively invariant region (i.e. invariant as time increases). By the 
well-known Poincar\'e-Bendixson theorem, the $\om$-limit set of any point in the region $R$  is a periodic orbit. Since the 
vertical line $AB$ can be moved to the right arbitrarily, the point $B$ can be moved up arbitrarily, and the closed curve 
$S$ can be arbitrarily small, the claim of the theorem is proved.
\end{proof}
\begin{remark}
As shown later, it can be verified numerically that all the $\om$-limit sets of points in the first open quadrant of the 
phase plane except the unstable steady state (\ref{fp}) is actually the same stable limit cycle. Proving such a claim is 
not easy. 
\end{remark}
\begin{theorem}
If all the $\om$-limit sets of points in the first open quadrant of the 
phase plane except the unstable steady state (\ref{fp}) is actually the same stable limit cycle $S_*$, then 
$S_*$ loops around the unstable steady state (\ref{fp}).
\end{theorem}
\begin{proof}
Let $W$ be a small tubular neighborhood of the limit cycle $S_*$, which is the local stable manifold of $S_*$, see 
Figure \ref{lca}. Let $S_1$ be a closed curve in the region R and near $S$ (Figure \ref{lcf}), that loops around the 
steady state $P_1$ once [in fact, $S_1$ can be just $S$]. For any $q \in S_1$, there is a segment neighborhood of $q$ 
in $S_1$, $\xi_q \subset S_1$ and a time $T_q$, such that $F^{T_q}(\xi_q) \subset W$ where $F^t$ is the evolution operator 
of the system (\ref{UMd1})-(\ref{UMd2}). All such segments form an open cover of $S_1$. By the compactness of $S_1$, 
there is a finite cover $\{ \xi_{q_n} \}_{n=1, \cdots, N}$. Let 
\[
T = \max_{n=1, \cdots, N} \{ T_{q_n} \} ,
\]
then $F^T(S_1) \subset W$. Since the region $R$ is positively invariant, $F^T(S_1)$ still loops around the steady state $P_1$,
then the tubular neighborhood $W$ also loops around the steady state $P_1$, thus the limit cycle $S_*$ also loops around 
the steady state $P_1$.
\end{proof}

Numerically one can verifies that the attractor in the first open quadrant of the phase plane is a limit cycle as shown 
in Figure \ref{flc3}. As $H$ is decreased, the limit cycle is quickly getting closer and closer to the $v$-axis (and 
$u$-axis) as shown in Figure \ref{flc41}. The time series graph of the limit cycle in Figure \ref{flc41}(b) is shown in 
Figure \ref{flc2}. One can see clearly that with small random environmental perturbations, the system will become extinct!

\section{Spatial Dependence}

\subsection{The Limit Cycle Is Linearly Stable}

Let $S_*$ be the limit cycle on the plane,
\[
S_*: \quad  u= u^*(t), \quad  v= v^*(t).
\]
The period of $S_*$ is $T_*$. Linearizing (\ref{Md1})-(\ref{Md2}) at $S_*$, we get 
\begin{eqnarray}
& & \frac{\pa u}{\pa t}  =  \frac{\pa^2 u}{\pa x^2} + \left [ 1-2u^*(t) - \frac{Hv^*(t)}{(u^*(t)+H)^2} \right ] u
 - \frac{u^*(t)}{u^*(t)+H} v , \label{llc1} \\
& & \frac{\pa v}{\pa t}  =  \frac{\pa^2 v}{\pa x^2} + \frac{kHv^*(t)}{(u^*(t)+H)^2}u + k \left [ \frac{u^*(t)}{u^*(t)+H}
-r \right ] v . \label{llc2} 
\end{eqnarray}
Using the Fourier mode 
\[
u = u_\xi e^{i\xi x} + \text{ c.c. }, \quad v = v_\xi e^{i\xi x} + \text{ c.c. },
\]
the linear system (\ref{llc1})-(\ref{llc2}) is transformed into 
\begin{eqnarray*}
& & \frac{\pa u_\xi}{\pa t}  = -\xi^2 u_\xi + \left [ 1-2u^*(t) - \frac{Hv^*(t)}{(u^*(t)+H)^2} \right ] u_\xi
 - \frac{u^*(t)}{u^*(t)+H} v_\xi , \\
& & \frac{\pa v_\xi}{\pa t}  = -\xi^2 v_\xi + \frac{kHv^*(t)}{(u^*(t)+H)^2}u_\xi + k \left [ \frac{u^*(t)}{u^*(t)+H}
-r \right ] v_\xi . 
\end{eqnarray*}
A further change of variables
\[
u_\xi = e^{-\xi^2t} \hu_\xi , \quad v_\xi = e^{-\xi^2t} \hv_\xi ,
\]
transforms this linear system into
\begin{eqnarray*}
& & \frac{\pa \hu_\xi}{\pa t}  = \left [ 1-2u^*(t) - \frac{Hv^*(t)}{(u^*(t)+H)^2} \right ] \hu_\xi
 - \frac{u^*(t)}{u^*(t)+H} \hv_\xi , \\
& & \frac{\pa \hv_\xi}{\pa t}  = \frac{kHv^*(t)}{(u^*(t)+H)^2}\hu_\xi + k \left [ \frac{u^*(t)}{u^*(t)+H}
-r \right ] \hv_\xi ;
\end{eqnarray*}
which is a stable linear system since the limit cycle is linearly stable in the plane. Since the plane is a subspace 
of the infinite dimensional phase space under Neumann or periodic boundary condition, the limit cycle is linearly stable in such 
an infinite dimensional phase space.

\subsection{The Riddled Basin of Attraction of the Limit Cycle}

Since it is linearly stable, the limit cycle on the plane is an attractor in the entire infinite dimensional phase space.
Numerical simulations are conducted on the system (\ref{Md1})-(\ref{Md2}) under the periodic boundary condition 
with spatial period $L=300$, where $c_1=0$, $c_2=0$, $r=0.8$, $H=0.1$, and $k=1.0$. The initial conditions of the 
numerical simulations are given by 
\[
u(x)= 0.4018, \  v(x)=0.3754 + \epsilon \cos(\pi x/L),
\]
for different values of $\e$. When $\e =0$, the initial condition lies on the limit cycle on the plane. Figure \ref{RB} 
illustrates the riddled nature of the basin of attraction of the limit cycle on the plane. This riddled nature indicates that  
spatial perturbations sometimes (but neither always nor never) can remove the paradox of enrichment near the limit cycle 
on the plane.

\subsection{Other Bifurcations?}

Linearizing (\ref{Md1})-(\ref{Md2}) at the steady state (\ref{fp}), we get 
\begin{eqnarray}
& & \frac{\pa u}{\pa t}  =  \frac{\pa^2 u}{\pa x^2} + r \left [ (1-2u_*-H)u-v \right ] , \label{glu1} \\
& & \frac{\pa v}{\pa t}  =  \frac{\pa^2 v}{\pa x^2} + rkH \frac{1-u_*}{u_*}u . \label{glu2} 
\end{eqnarray}
Using the Fourier mode 
\[
u = u_\xi e^{i\xi x} + \text{ c.c. }, \quad v = v_\xi e^{i\xi x} + \text{ c.c. },
\]
the linear system (\ref{glu1})-(\ref{glu2}) is transformed into 
\begin{eqnarray*}
& & \frac{\pa u_\xi}{\pa t}  =  -\xi^2 u_\xi + r \left [ (1-2u_*-H)u_\xi-v_\xi \right ] , \\
& & \frac{\pa v}{\pa t}  =  -\xi^2 v_\xi + rkH \frac{1-u_*}{u_*}u_\xi . 
\end{eqnarray*}
The eigenvalues $\la$ of this system satisfy 
\begin{equation}
\la^2 +\left [ 2\xi^2 -r(1-2u_*-H) \right ] \la +\xi^2 \left [ \xi^2  -r(1-2u_*-H) \right ] 
+ r^2 kH \frac{1-u_*}{u_*} = 0 .
\label{gss}
\end{equation}
A possible Hopf bifurcation often occurs when the coefficient of the $\la$ term is zero, i.e.
\[
2\xi^2 -r(1-2u_*-H) = 0;
\]
while a possible Turing bifurcation often occurs when the constant term is zero, i.e.
\[
\xi^2 \left [ \xi^2  -r(1-2u_*-H) \right ] 
+ r^2 kH \frac{1-u_*}{u_*} = 0 .
\]
Notice also that when the the coefficient of the $\la$ term is zero, the following part (of the 
constant term) 
\[
\xi^2 \left [ \xi^2  -r(1-2u_*-H) \right ] 
\]
is minimal in $\xi^2$.

The minimum of $[ 2\xi^2 -r(1-2u_*-H)]$ in $\xi^2$ occurs at $\xi^2=0$, where the minimal value is 
$-r(1-2u_*-H)$. A Hopf bifurcation starts to occur at 
\begin{equation}
1-2u_*-H = 0 , \quad \text{i.e. } H = \frac{2}{1+r} -1 ; 
\label{hbc}
\end{equation}
as shown before in (\ref{H}). At this Hopf bifurcation, 
\[
\la = \pm i r \sqrt{kH \frac{1-u_*}{u_*}}.
\]
The steady state (\ref{fp}) bifurcates into the limit cycle (Figure \ref{flc3}) on the plane.

When $1-2u_*-H < 0$ (i.e. $H > \frac{2}{1+r} -1$), by (\ref{gss}), the 
steady state (\ref{fp}) is linearly stable in the infinite dimensional phase space. Therefore, further bifurcation 
may only occur when $1-2u_*-H > 0$ (i.e. $H < \frac{2}{1+r} -1$). When $H < \frac{2}{1+r} -1$, further Hopf bifurcations
may occur at
\begin{equation}
\xi^2 = \frac{1}{2}r(1-2u_*-H) = \frac{1}{2}r \left (1-\frac{1+r}{1-r}H \right ), 
\label{tbm}
\end{equation}
when 
\[
k >\frac{1}{4} \frac{r}{1-r-rH} \left [ 1-\frac{1+r}{1-r}H \right ]^2 .
\]
But numerically we did not observe such a bifurcation; the reason seems to be that the limit cycle in the plane is linearly stable 
for all these parameter values. 

As mentioned above, the minimum of $\xi^2 [ \xi^2  -r(1-2u_*-H)]$ in $\xi^2$ also occurs at (\ref{tbm}),
where the minimal value is
\[
-\frac{1}{4} r^2 (1-2u_*-H)^2 .
\]
A Turing bifurcation may start to occur at 
\[
\frac{1}{4} r^2 (1-2u_*-H)^2 = r^2 kH \frac{1-u_*}{u_*},
\]
that is
\begin{equation}
k = \frac{1}{4} \frac{r}{1-r-rH} \left [ 1-\frac{1+r}{1-r}H \right ]^2 ,
\label{tbc}
\end{equation}
where $\la^2=0$. When 
\[
k <\frac{1}{4} \frac{r}{1-r-rH} \left [ 1-\frac{1+r}{1-r}H \right ]^2 ,
\]
further Turing bifurcations may occur at 
\begin{eqnarray*}
\xi^2 &=&  \frac{1}{2}r(1-2u_*-H) \pm \sqrt{\left [ \frac{1}{2}r(1-2u_*-H) \right ]^2 -  r^2 kH \frac{1-u_*}{u_*}} \\
&=& \frac{1}{2}r \left (1-\frac{1+r}{1-r}H \right ) \pm \sqrt{\left [ \frac{1}{2}r \left (1-\frac{1+r}{1-r}H \right ) 
\right ]^2 - rk (1-r-rH)}. 
\end{eqnarray*}
Numerically we did not observe any Turing bifurcation. Usually Turing bifurcations are observed in (spatially) high 
dimensional systems or one dimensional systems with variable coefficients. 

\section{Conclusion}

In this article, we present a resolution on the paradox of enrichment based upon the dimensionless form of 
the mathematical model. We also explore spatial perturbations. The conclusion is that spatial perturbation 
sometimes (but neither always nor never) can remove the paradox of enrichment near the limit cycle 
on the plane.

\end{document}